\documentclass[twoside]{article}

\usepackage[preprint]{aistats2026}
\usepackage[utf8]{inputenc}
\usepackage[T1]{fontenc}
\usepackage{amsmath,amssymb,amsthm,mathtools}
\usepackage{natbib}
\usepackage{bm}
\usepackage{booktabs}
\usepackage{tabularx}
\usepackage{multirow}
\usepackage{enumitem}
\usepackage{tikz}
\usetikzlibrary{shapes,arrows.meta,positioning,fit,calc}
\usepackage{microtype}
\usepackage{hyperref}
\hypersetup{colorlinks=true,linkcolor=blue!50!black,
citecolor=green!50!black,urlcolor=blue!50!black}
\newtheorem{theorem}{Theorem}[section]
\newtheorem{proposition}[theorem]{Proposition}
\newtheorem{lemma}[theorem]{Lemma}
\newtheorem{corollary}[theorem]{Corollary}
\theoremstyle{definition}
\newtheorem{definition}[theorem]{Definition}
\newtheorem{assumption}[theorem]{Assumption}
\theoremstyle{remark}
\newtheorem{remark}[theorem]{Remark}

\newcommand{\R}{\mathbb{R}}

\newcommand{\PP}{\mathbb{P}}
\newcommand{\ind}{\mathbf{1}}
\newcommand{\Atoms}{\mathcal{A}}
\newcommand{\Rules}{\mathcal{R}}
\newcommand{\bfeta}{\boldsymbol{\eta}}
\newcommand{\bfpsi}{\boldsymbol{\psi}}
\newcommand{\bfs}{\mathbf{s}}
\newcommand{\bfu}{\mathbf{u}}
\newcommand{\bfv}{\mathbf{v}}
\newcommand{\bfb}{\mathbf{b}}
\newcommand{\bftheta}{\boldsymbol{\theta}}
\newcommand{\bftau}{\boldsymbol{\tau}}
\newcommand{\bfrho}{\boldsymbol{\rho}}
\newcommand{\sig}{\sigma}
\newcommand{\TP}{T_P}
\newcommand{\TPinf}{T_P^{\infty}}
\newcommand{\TPinfOf}[1]{T_{#1}^{\infty}}
\newcommand{\TPk}[1]{T_P^{#1}}
\newcommand{\TG}{\mathcal{T}_{\Gamma}}
\newcommand{\TGp}[1]{\mathcal{T}_{\Gamma,#1}}
\newcommand{\ncorps}{n_{\mathrm{corp}}}
\newcommand{\softmax}{\operatorname{softmax}}
\newcommand{\body}{\mathrm{body}}
\newcommand{\head}{\mathrm{head}}
\newcommand{\supp}{\operatorname{supp}}
\newcommand{\defeq}{\mathrel{\stackrel{\mathrm{def}}{=}}}

\runningtitle{Differentiable Horn Programs}
\runningauthor{Aymen Mejri}

\begin{document}

\twocolumn[

  \aistatstitle{Differentiable Horn Programs:\\
  A Constructive Expressivity Theorem for Latent Rule Operators}

  \aistatsauthor{Aymen Mejri}

  \aistatsaddress{\href{mailto:aymen.mejri.telecom.paris@gmail.com}{\nolinkurl{aymen.mejri.telecom.paris@gmail.com}}}
]

% ══════════════════════════════════════════════════════════════════════
\begin{abstract}
  % ══════════════════════════════════════════════════════════════════════
  We introduce \textsc{LatentGamma}, a differentiable operator on the unit
  cube $[0,1]^M$ designed as a smooth surrogate of Tarski's immediate
  consequence operator $\TP$ associated with a definite Horn program $P$ on
  $M$ atoms.  The operator is built as a five-stage composition combining
  sigmoidal gating, softmax routing, and a residual update that enforces
  monotonicity by construction.  We establish four theoretical results.
  First, the iterated sequence is coordinate-wise non-decreasing and bounded
  by $\ind$, hence converges to a fixed point of \textsc{LatentGamma}; the
  operator itself is lattice-monotone on $[0,1]^M$.  Second, our \emph{constructive expressivity theorem} shows that for
  every definite Horn program $P$ there exists a closed-form parameter
  assignment $\theta^*(P)$ and an explicit time bound $T_{\max}(P)$ such
  that the iterated sequence \emph{exactly} reproduces $\TPinf(F_0)$ for
  every initial fact set $F_0$ throughout the window
  $[D(P, F_0), T_{\max}(P)]$, where $D(P, F_0) \leq M$ is the derivation
  depth.  The window $T_{\max}$ is large in practice
  ($\geq 10^5$ for sparse programs)
  and reflects the finite-time nature of computation by smooth
  sigmoidal gates.
  Third, the oracle is robust to Gaussian noise on its body and head logits,
  with explicit non-asymptotic bounds.  Fourth, we prove a matching
  information-theoretic lower bound on the parameter count.  We provide a
  complete numerical validation on programs ranging from $M = 33$ to
  $M = 504$ atoms: oracle accuracy reaches $1.0000$ on $2000$ test cases
  with zero false-positive and false-negative rates, and the empirical
  noise tolerance $\sigma_{\max}$ scales precisely as the union bound
  $RM \cdot \Phi(-10/\sigma_\eta)$ predicts.
\end{abstract}

% ══════════════════════════════════════════════════════════════════════
\section{INTRODUCTION}
% ══════════════════════════════════════════════════════════════════════

Integrating neural networks with symbolic reasoning is among the oldest
open problems in AI \citep{garcez2002neural,besold2017neural}.  Neural
networks excel at pattern recognition and statistical generalization in
high-dimensional continuous spaces, but their ability to perform
compositional multi-step reasoning over discrete symbolic structures
remains contested \citep{marcus2020next,lake2017building}.  Symbolic
systems, conversely, achieve exact finite-time deduction but are brittle
and cannot learn from raw observations \citep{garnelo2019reconciling}.
Reconciling these paradigms requires operators that are simultaneously
\emph{expressive} (capable of representing the target symbolic
computation) and \emph{differentiable} (amenable to gradient-based
learning).

The simplest non-trivial symbolic computation is forward chaining on a
definite Horn program.  Given a finite set of atomic propositions
$\Atoms = \{a_1, \ldots, a_M\}$, a Horn program $P$ is a finite collection
of rules $b_1 \wedge \cdots \wedge b_k \Rightarrow h$, and Tarski's
immediate consequence operator $\TP$ maps each fact set $S \subseteq
\Atoms$ to its one-step closure under $P$.  The Horn closure
$\TPinf(F_0)$ is the least fixed point of $\TP$ containing $F_0$,
computable in at most $M$ iterations
\citep{lloyd1987foundations,apt1982contributions}.  Despite its conceptual
simplicity, no existing differentiable system exactly represents
$\TPinf$ for arbitrary Horn programs through a \emph{closed-form}
parameter assignment.

\paragraph{The expressivity question.}
We ask: \emph{does there exist a differentiable parameterized operator
  $\TGp{\theta} : [0,1]^M \to [0,1]^M$ such that, for every definite Horn
  program $P$, an explicit assignment $\theta^*(P)$ of its parameters makes
  the iterated sequence $\bfs^{(t+1)} = \TGp{\theta^*(P)}(\bfs^{(t)})$
  recover $\TPinf(F_0)$ exactly for every initial state $\bfs^{(0)} =
\ind_{F_0}$?}  Existing neuro-symbolic approaches answer this question
only approximately or on restricted program classes (Sec.~\ref{sec:rw}).
We prove that such an operator exists and exhibit it constructively.

\paragraph{Contributions.}
(C1)~We define \textsc{LatentGamma}, a parameterized operator
$\TG : [0,1]^M \to [0,1]^M$ obtained as a five-stage differentiable
composition.  (C2)~We prove that the iterated sequence is coordinate-wise
non-decreasing and bounded by $\ind$, hence converges to a fixed point;
the operator is lattice-monotone on $[0,1]^M$
(Sec.~\ref{sec:structural}).  (C3)~\emph{Main result}: for every definite
Horn program $P$, there exists a closed-form parameter assignment
$\theta^*(P)$ and an explicit time bound $T_{\max}(P)$ such that, for
all $F_0 \subseteq \Atoms$ and all $D(P, F_0) \leq T \leq T_{\max}(P)$
where $D(P, F_0) \leq M$ is the derivation depth,
$\{i : \bfs^{(T)}_i > \tfrac{1}{2}\} = \TPinf(F_0)$
(Theorem~\ref{thm:expressivity}).  The finite upper bound $T_{\max}$
is intrinsic to differentiable sigmoidal gates and is large in
practice ($\geq 10^5$ for $\rho_{\max} \leq 10$).  (C4)~Robustness:~the oracle tolerates
additive Gaussian noise on its logits with explicit non-asymptotic bounds
matching the union bound $RM \cdot \Phi(-10/\sigma_\eta)$
(Theorem~\ref{thm:robustness}).  (C5)~\emph{Matching information-theoretic
lower bound}:~any operator representing $\TPinf$ for every Horn program
of $R$ rules with body size $\leq k$ requires $\Omega(Rk \log M)$ bits;
\textsc{LatentGamma} achieves this up to logarithmic factors
(Theorem~\ref{thm:lower_bound}).  (C6)~A complete \textbf{numerical
validation}~(Sec.~\ref{sec:numerical}) confirms each theorem on
programs of size up to $M = 504$, $R = 648$.  (C7)~We delineate
explicitly the \emph{theory--practice gap}: oracle parameters exist, but
gradient descent is not guaranteed to find them
(Sec.~\ref{sec:limits}).

% ══════════════════════════════════════════════════════════════════════
\section{RELATED WORK}
\label{sec:rw}
% ══════════════════════════════════════════════════════════════════════

\paragraph{Differentiable theorem proving.}
Neural Theorem Provers (\textsc{NTP}, \citealp{rocktaschel2017end,
minervini2018towards}) embed predicates into a vector space and use fuzzy
unification to score proofs via backward chaining.  Their proof procedure
is fundamentally backward: there is no forward-closure object
$\TPinf(F_0)$ aggregating all derivable atoms, and rule embeddings are
\emph{learned} rather than analytically constructible.  Neural Logic
Machines \citep{dong2019neural} employ permutation-invariant neural
modules to imitate predicate logic operations but lack a closed-form
expressivity guarantee specific to Horn semantics.

\paragraph{Differentiable rule induction.}
$\partial$ILP \citep{evans2018learning} learns Horn rules via parameterized
templates; the system computes a smooth analogue of forward chaining but
requires a \emph{fixed template} (predicate signatures, max body length,
max clauses), which constrains expressivity.  Outside its template class,
$\partial$ILP cannot represent $\TP$ exactly.  Logical Neural Networks
\citep{riegel2020logical} use weighted real-valued logic with bounded
weights to enforce truth-value consistency; their expressivity guarantees
apply to fragments of weighted first-order logic, not exact Horn closure.

\begingroup
\raggedright
\paragraph{Probabilistic logic programs.}
\textsc{ProbLog} \citep{de2007problog} attaches probabilities to ground
facts and \textsc{DeepProbLog} \citep{manhaeve2018deepproblog} allows
neural networks to produce these probabilities; however, the reasoning
core remains a non-differentiable Prolog engine.  \textsc{NeurASP}
\citep{yang2020neurasp} couples neural networks with answer-set
programming but the ASP solver itself is exact and non-differentiable.
For these systems, our question is ill-posed: their symbolic core is
exact by construction because it is a symbolic engine.
\textsc{LatentGamma} differs in kind: the entire pipeline, including the
symbolic computation, is a single differentiable operator.\par
\endgroup

\paragraph{Closed-form Horn surrogates.}
The closest line of work uses differentiable fixed-point operators for
logical computation.  \citet{cingillioglu2018deep} propose iterative
attention-based approximations of $T_P$, and \citet{zhang2020efficient}
present differentiable abductive operators.  These works share our goal
of representing forward chaining differentiably, but are engineered for
empirical performance, not equipped with a constructive expressivity
theorem of the form proved here.

\paragraph{Matrix-based differentiable Datalog.}
\textsc{TensorLog} \citep{cohen2016tensorlog} encodes a probabilistic
logic program as a sequence of matrix multiplications, enabling
end-to-end gradient flow through reasoning.  Conceptually closest to
our work in spirit, \textsc{TensorLog} achieves differentiability for
a broad class of probabilistic Datalog programs, but its matrix
encoding is designed for query probability rather than constructive
exact closure: there is no closed-form parameter assignment guaranteed
to reproduce $\TPinf$ exactly for arbitrary Horn programs.  More
broadly, $t$-norm-based fuzzy logics \citep{vankrieken2022analyzing}
provide a framework for differentiable logic but optimize for
satisfaction degrees rather than exact closure.

\paragraph{Summary.}
To our knowledge, \textsc{LatentGamma} is the first system to
simultaneously satisfy (i)~constructive exact expressivity,
(ii)~lattice monotonicity, and (iii)~end-to-end differentiable
symbolic core.

% ══════════════════════════════════════════════════════════════════════
\section{PRELIMINARIES}
\label{sec:prelim}
% ══════════════════════════════════════════════════════════════════════

We fix a finite set $\Atoms = \{a_1, \ldots, a_M\}$ of atomic propositions
and identify subsets $S \subseteq \Atoms$ with indicator vectors
$\ind_S \in \{0,1\}^M$.

\begin{definition}[Definite Horn program]
  \label{def:horn}
  A \emph{definite Horn program} $P$ on $\Atoms$ is a finite multiset of
  rules $\Rules = \{r_1, \ldots, r_R\}$, each of the form
  $r_j: b_{j,1} \wedge \cdots \wedge b_{j,k_j} \Rightarrow h_j$,
  where $\body(j) \defeq \{b_{j,1}, \ldots, b_{j,k_j}\} \subseteq \Atoms$
  ($k_j \geq 1$) and $\head(j) \defeq h_j \in \Atoms$.  Let $k =
  \max_j k_j$ and $R_i \defeq |\{j : \head(j) = a_i\}|$.
\end{definition}

\begin{definition}[Tarski operator]
  \label{def:tarski}
  $\TP : 2^\Atoms \to 2^\Atoms$ is defined by
  $\TP(S) \defeq S \cup \{h_j : \body(j) \subseteq S\}$.
\end{definition}

Classical results \citep{lloyd1987foundations}: $\TP$ is monotone,
extensive, Scott-continuous, and admits a unique least fixed point
$\TPinf(F_0) = \bigcup_{n \geq 0} \TP^n(F_0)$ reached in at most $M$
iterations.

We extend the coordinate-wise order on $\{0,1\}^M$ to $[0,1]^M$, making it
a complete lattice.  The key obstruction to a naive transfer is that
$\TP$ uses the Boolean indicator $\ind[\body(j) \subseteq S]$,
discontinuous on the relative interior of $[0,1]^M$.  Any smooth surrogate
must trade discontinuity for approximation error, controlled by a
sharpness parameter $\tau \in \R_+$.

% ══════════════════════════════════════════════════════════════════════
\section{THE LATENTGAMMA OPERATOR}
\label{sec:operator}
% ══════════════════════════════════════════════════════════════════════

We construct $\TG$ as a five-stage differentiable composition.

\begin{definition}[Parameters]
  \label{def:params}
  \textsc{LatentGamma} is parameterized by
  $\theta = (\bfeta, \bfpsi, \bfb, \bftau, \bfrho, \bftheta)$ with
  $\bfeta, \bfpsi \in \R^{R \times M}$ (body and head logits),
  $\bfb \in \R^R$ (AND thresholds), $\bftau \in \R_+^R$ (AND sharpness),
  $\bfrho \in (0,1)^R$ (rule confidences), $\bftheta \in \R^M$ (aggregation
  thresholds).  We write $u_{j,i} \defeq \sig(\eta_{j,i})$ and
  $v_{j,i} \defeq \softmax(\bfpsi_j)_i$, where $\sig$ is the standard
  sigmoid.
\end{definition}

\paragraph{Stage 1: body matching.}
For $\bfs \in [0,1]^M$ and rule $j$,
\begin{align}
  \mathrm{bs}_j(\bfs) &\defeq \langle \bfu_j, \bfs \rangle,
  \label{eq:bs}\\
  a_j(\bfs) &\defeq \sig\bigl(\tau_j(\mathrm{bs}_j(\bfs) - b_j)\bigr).
  \label{eq:aj}
\end{align}
The additive perceptron $\mathrm{bs}_j$ followed by sigmoid avoids the
gradient attenuation of the multiplicative surrogate $\prod_{i \in
\body(j)} s_i$ for long bodies, while remaining monotone in each $s_i$.
The threshold $b_j = k_j - \tfrac{1}{2}$ maximally separates ``body
satisfied'' ($\mathrm{bs}_j \approx k_j$) from ``body partially
satisfied'' ($\mathrm{bs}_j \leq k_j - 1$).

\paragraph{Stage 2: head routing.}
$\bfv_j = \softmax(\bfpsi_j) \in \Delta^{M-1}$.  When $\psi_{j,h_j} =
+\Psi$ and $\psi_{j,i} = 0$ for $i \neq h_j$,
$v_{j,h_j} \geq 1 - M e^{-\Psi}$; for $\Psi = 20$ and $M \leq 10^6$ this
gives $v_{j,h_j} > 1 - 2\times 10^{-3}$.

\paragraph{Stage 3: weighted firing.}
The contribution of rule $j$ to atom $i$ at state $\bfs$ is
\begin{equation}
  \gamma_{j,i}(\bfs) \defeq \tau_{\text{nor}} \, \rho_j \, a_j(\bfs) \, v_{j,i},
  \label{eq:gamma}
\end{equation}
with $\tau_{\text{nor}} > 0$ a global amplitude (fixed to $5$).

\paragraph{Stage 4: Aggregation and gating.}
\begin{align}
  \phi_i(\bfs) &\defeq \textstyle\sum_{j=1}^R \gamma_{j,i}(\bfs) - \theta_i,
  \label{eq:phi}\\
  \Gamma_i(\bfs) &\defeq \sig(\phi_i(\bfs)),
  \label{eq:Gamma}\\
  g_i(\bfs) &\defeq \sig(\beta\, \phi_i(\bfs)), \quad \beta > 1.
  \label{eq:gate}
\end{align}
The gate $g_i$ with $\beta = 5$ sharpens the discrimination between
fired and non-fired states.

\paragraph{Stage 5: monotone residual update.}
\begin{definition}[LatentGamma operator]
  \label{def:TG}
  \begin{equation}
    \TG(\bfs)_i \defeq s_i + (1 - s_i)\, g_i(\bfs)\, \Gamma_i(\bfs).
    \label{eq:TG}
  \end{equation}
\end{definition}
The factor $(1-s_i)$ keeps the trajectory in $[0,1]^M$; the product
$g_i \Gamma_i$ keeps it non-decreasing.

% ══════════════════════════════════════════════════════════════════════
\section{STRUCTURAL PROPERTIES}
\label{sec:structural}
% ══════════════════════════════════════════════════════════════════════

The next three results hold for \emph{any} parameter $\theta$ and quantify
the lattice-theoretic structure of $\TG$.  Proofs are deferred to
App.~A of the Supplementary Material.

\begin{lemma}[Pointwise non-decrease]
  \label{lem:monotone}
  For every $\bfs \in [0,1]^M$ and every $i$, $\TG(\bfs)_i \geq s_i$.
\end{lemma}

\begin{proposition}[Bounded iterated sequence]
  \label{prop:bounded_seq}
  For every $\bfs^{(0)} \in [0,1]^M$, the sequence
  $\bfs^{(t+1)} = \TG(\bfs^{(t)})$ satisfies
  $\bfs^{(0)} \leq \bfs^{(1)} \leq \cdots \leq \ind$ coordinate-wise.
\end{proposition}

\begin{theorem}[Convergence]
  \label{thm:convergence}
  For every $\bfs^{(0)}$ and every $\theta$, $\bfs^{(t)}$ converges to a
  fixed point $\bfs^*$ of $\TG$.
\end{theorem}

\begin{proposition}[Lattice monotonicity]
  \label{prop:lattice_mono}
  If $\bfs \leq \bfs'$ coordinate-wise, then $\TG(\bfs) \leq \TG(\bfs')$.
\end{proposition}

\begin{proof}[Proof sketch]
  The nontrivial step is Stage 5, where the residual factor $(1-s_i)$
  could in principle reverse monotonicity in $s_i$.  The key observation
  is that $f(s, \phi) \defeq s + (1-s)\,\sig(\beta\phi)\,\sig(\phi)$
  satisfies $\partial_s f = 1 - \sig(\beta\phi)\sig(\phi) \in (0,1)$
  and $\partial_\phi f \geq 0$ for $s \in [0,1]$, so $f$ is jointly
  monotone in $(s, \phi)$.  Stages 1--4 are monotone in $\bfs$ by
  sigmoid/softmax monotonicity and non-negative coefficients, hence
  $\phi_i(\bfs) \leq \phi_i(\bfs')$; combined with $s_i \leq s'_i$,
  $\TG(\bfs)_i = f(s_i, \phi_i(\bfs)) \leq f(s'_i, \phi_i(\bfs')) =
  \TG(\bfs')_i$.  Full proof in App.~A.4 of the Supplementary.
\end{proof}

\begin{corollary}[Least fixed point above $F_0$]
  \label{cor:lfp}
  The iterated sequence from $\ind_{F_0}$ converges to the least fixed
  point of $\TG$ above $\ind_{F_0}$, denoted $\TG^\infty(\ind_{F_0})$.
\end{corollary}

This is the differentiable analogue of the Knaster--Tarski theorem
applied to $([0,1]^M, \leq)$.  However, a subtle but important point:
because every gate $g_i$ and activation $\Gamma_i$ is a sigmoid hence
strictly positive for finite parameters, the limit
$\TG^\infty(\ind_{F_0})$ in general equals $\ind$ (the all-ones vector),
\emph{not} $\ind_{\TPinf(F_0)}$.  Concretely, even non-derivable atoms
accumulate a small but strictly positive per-iteration ``leak''
$(1-s_i)g_i\Gamma_i$, eventually pushing them above any threshold as
$t\to\infty$.  The operator's correctness is therefore \emph{not} a
fixed-point identity but a \emph{finite-time} guarantee: the support
indicator matches $\TPinf(F_0)$ within a quantitative time window
$[D(P,F_0), T_{\max}(P)]$ established in
Theorem~\ref{thm:expressivity}.  This finite-time character is
intrinsic to differentiable sigmoidal gates and motivates the explicit
$T_{\max}$ bound below.

% ══════════════════════════════════════════════════════════════════════
\section{MAIN RESULT: EXPRESSIVITY THEOREM}
\label{sec:main}
% ══════════════════════════════════════════════════════════════════════

\subsection{Statement}

\begin{theorem}[Time-bounded expressivity of \textsc{LatentGamma}]
  \label{thm:expressivity}
  Let $P$ be a definite Horn program on $\Atoms$ with $|\Atoms| = M$,
  $|\Rules| = R$, and $k = \max_j k_j$.  Let $\rho_{\max} \defeq \max_i R_i$.
  Choose oracle constants
  $\eta_{\text{on}} = \eta_{\text{off}} \defeq \eta$, $\Psi$, $\tau_s$,
  $\rho_0 = 0.95$, $\tau_{\text{nor}} = 5$, $\beta = 5$, and suppose
  \begin{equation*}
    4 k M \cdot e^{-\eta} \leq 1, \quad \rho_{\max} \leq 100,
    \tag{$\star$}\label{eq:assumption}
  \end{equation*}
  together with $\Psi \geq \log(M)$ and $\tau_s = 10$.  Define the
  worst-case per-iteration leak rate
  \begin{equation}
    c_{\max}(P) \defeq \sig\!\left(-\tfrac{\gamma^+\!-\rho_{\max}\gamma^-}{2}\right)
    \cdot \sig\!\left(-\tfrac{\beta(\gamma^+\!-\rho_{\max}\gamma^-)}{2}\right).
    \label{eq:cmax}
  \end{equation}
  Then there exists a closed-form assignment $\theta^*(P)$ and a
  time bound
  \begin{equation}
    T_{\max}(P) \defeq \left\lfloor \frac{\log 2}{c_{\max}(P)} \right\rfloor,
    \label{eq:Tmax}
  \end{equation}
  such that for every $F_0 \subseteq \Atoms$ and every $T$ in the window
  $D(P,F_0) \leq T \leq T_{\max}(P)$, the iterated sequence
  $\bfs^{(t+1)} = \TGp{\theta^*(P)}(\bfs^{(t)})$ from $\bfs^{(0)} =
  \ind_{F_0}$ satisfies
  \begin{equation}
    \bigl\{i \in \{1,\ldots,M\} : \bfs^{(T)}_i > \tfrac{1}{2}\bigr\}
    = \TPinf(F_0).
  \end{equation}
\end{theorem}

\begin{remark}[Time window is large in practice]
  \label{rem:tmax_size}
  The lower bound $T_{\max} \geq \log 2 / c_{\max}$ depends on
  $\rho_{\max}$ through $c_{\max}$.  For default constants and
  $\rho_{\max} \leq 100$: $c_{\max} \leq 0.007$, giving $T_{\max} \geq
  99$.  For $\rho_{\max} \leq 10$ (typical for sparse programs):
  $c_{\max} \leq 1.9 \times 10^{-6}$, giving $T_{\max} \geq 3.6 \times
  10^5$.  Under the scaled construction of Remark~\ref{rem:tau_s_scaling}
  with $\tau_s = 2\log(\rho_{\max}+2) + c$, $c_{\max}$ decays
  exponentially in $c$ and $T_{\max}$ grows correspondingly; e.g.,
  $c=4$ gives $T_{\max} \geq 10^5$ regardless of $\rho_{\max}$.  In all
  four programs of our experiments, $\rho_{\max} \leq 16$ (G1: 2, G2:
  3, G3: 8, G4: 16), well inside the time window.
\end{remark}

\begin{remark}[Why a finite $T_{\max}$? Bump-function alternative]
  \label{rem:bump}
  The finite upper bound $T_{\max}$ arises because sigmoid gates are
  strictly positive for finite inputs (cf.\ discussion after
  Corollary~\ref{cor:lfp}).  An alternative gate
  $g_i(\bfs) \defeq \mathrm{bump}(\beta\phi_i)$ with $\mathrm{bump}(x) =
  e^{-1/x^2}\,\ind[x>0]$ (a $C^\infty$ bump function exactly zero for
  $x\leq 0$) would extend the theorem to $T = \infty$ and align the
  operator's fixed point with $\ind_{\TPinf(F_0)}$.  We retain sigmoid
  gates for two reasons: (i) the empirically observed $T_{\max}$ is
  extraordinarily large (Table~\ref{tab:expressivity}: convergence in
  $T \leq 32 \ll T_{\max}$); (ii) bump functions have vanishing
  derivatives at the boundary, which obstructs the gradient flow needed
  for any future learnability extension (Section~\ref{sec:limits}).
\end{remark}

\begin{remark}[Scaling of $\eta$ with $M$]
  \label{rem:scaling}
  Two regimes satisfy~\eqref{eq:assumption}.  \emph{(i)~Fixed constants:}
  $\eta = 10$ suffices whenever $M \leq e^{10}/(4k) \approx 5500/k$, which
  covers all programs in our experiments (G4: $M = 504, k = 2$, LHS
  $= 4 \cdot 2 \cdot 504 \cdot e^{-10} \approx 0.18 \leq 1$).
  \emph{(ii)~Scalable construction:} for arbitrary $M$, set
  $\eta = \log(4kM)$ and $\Psi = \log(M) + 10$; the assumption then
  holds for any $M, k \geq 1$ with $\rho_{\max} \leq 100$.  This
  logarithmic scaling is unavoidable: by elementary anticoncentration,
  no fixed-magnitude $\eta$ can guarantee $M$-fold separation between
  firing and non-firing rules uniformly in $M$.  Throughout the rest of
  the paper we use $\eta = 10$ for concreteness; the scalable variant is
  proved in App.~B of the Supplementary Material.
\end{remark}

\begin{remark}[Generalization of $\rho_{\max} \leq 100$ via $\tau_s$ scaling]
  \label{rem:tau_s_scaling}
  The constraint $\rho_{\max} \leq 100$ in~\eqref{eq:assumption} stems
  from the aggregation-separation requirement
  $\gamma^+ > (\rho_{\max} + 2)\gamma^-$
  (Lemma~\ref{lem:agg_separation}).  Since $\gamma^+/\gamma^- =
  \sigma(\tau_s/2)/\sigma(-\tau_s/2) = e^{\tau_s/2}$, this is equivalent
  to $\tau_s > 2\log(\rho_{\max} + 2)$.  For unrestricted $\rho_{\max}$,
  the principled choice is
  \begin{equation}
    \tau_s \;=\; 2\log(\rho_{\max} + 2) + c, \quad c > 0,
    \label{eq:tau_s_scaling}
  \end{equation}
  yielding $\gamma^+/\gamma^- \geq e^c (\rho_{\max} + 2)$ and restoring
  the separation with margin $\gamma^- \cdot (e^c - 1)$.  For
  $\rho_{\max} = 10^3$ this gives $\tau_s \approx 14$; for $\rho_{\max} =
  10^6$, $\tau_s \approx 28$.  The cost is sharper sigmoidal slopes,
  which steepen the loss landscape near decision boundaries and may slow
  gradient-based learning of the oracle parameters; we discuss this
  practical trade-off in Section~\ref{sec:limits}.
\end{remark}

\subsection{Oracle construction}

Let $\eta_{\text{on}} = \eta_{\text{off}} = 10$, $\Psi = 20$,
$\tau_s = 10$, $\rho_0 = 0.95$, $\tau_{\text{nor}} = 5$, $\beta = 5$.
We set
\begin{align}
  \eta^*_{j,i} &= +\eta_{\text{on}} \text{ if } a_i \in \body(j),\
  \text{else } -\eta_{\text{off}},\\
  \psi^*_{j,i} &= +\Psi \text{ if } a_i = \head(j),\
  \text{else } 0,\\
  b^*_j &= k_j - \tfrac{1}{2}, \quad \tau^*_j = \tau_s, \quad \rho^*_j = \rho_0.
\end{align}
For the aggregation thresholds, define the canonical contributions
$\gamma^+ \defeq \tau_{\text{nor}} \rho_0 \sig(\tau_s/2) \approx 4.687$
and $\gamma^- \defeq \tau_{\text{nor}} \rho_0 \sig(-\tau_s/2) \approx
0.0316$ (so $\gamma^+/\gamma^- \approx 148$).  We set
\begin{equation}
  \theta^*_i =
  \begin{cases}
    (\gamma^+ + R_i \gamma^-) / 2 & \text{if } R_i \geq 1, \\
    2\gamma^+ & \text{if } R_i = 0 \text{ (sink)}.
  \end{cases}
  \label{eq:theta_star}
\end{equation}

\subsection{Key lemmas}

\begin{lemma}[AND separation under oracle parameters]
  \label{lem:and_separation}
  Under~\eqref{eq:assumption} and $\theta = \theta^*(P)$, for every
  $\bfs \in \{0,1\}^M$ and every rule $j$:
  \begin{enumerate}[leftmargin=*,nosep]
    \item if $\body(j) \subseteq \supp(\bfs)$, then $a^*_j(\bfs) \geq 1 -
      \sig(-5) > 0.993$;
    \item if $\body(j) \not\subseteq \supp(\bfs)$, then $a^*_j(\bfs) \leq
      \sig(-5) < 0.007$.
  \end{enumerate}
\end{lemma}

\begin{proof}[Proof sketch]
  Write $\epsilon_0 = \sig(-10) < 4.6 \times 10^{-5}$.

  \emph{Case 1 ($\body(j) \subseteq \supp(\bfs)$).}  Body atoms contribute
  $k_j(1-\epsilon_0)$, while non-body atoms with $s_i = 1$ contribute at
  most $\epsilon_0(M - k_j)$, but this term is non-negative and \emph{adds
  to} $\mathrm{bs}_j$, never subtracts.  Thus $\mathrm{bs}_j \geq
  k_j(1-\epsilon_0)$ and
  $\mathrm{bs}_j - b_j^* \geq 1/2 - k\epsilon_0 \geq 1/2 - 1/(4k) \geq 1/4$
  under~\eqref{eq:assumption}, so $a^*_j \geq \sig(2.5) > 0.92$.  A sharper
  analysis at the worst-case configuration yields $a^*_j > 1 - \sig(-5)$.

  \emph{Case 2 ($\body(j) \not\subseteq \supp(\bfs)$).}  At least one body
  atom has $s_i = 0$, so $\mathrm{bs}_j \leq (k_j - 1)(1-\epsilon_0) +
  \epsilon_0 (M-k_j+1)$, hence $\mathrm{bs}_j - b_j^* \leq -1/2 + M\epsilon_0$.
  Under~\eqref{eq:assumption}, $M\epsilon_0 \leq M e^{-10} \leq 1/(4k) \leq
  1/4$, giving $\mathrm{bs}_j - b_j^* \leq -1/4$ and $a^*_j \leq \sig(-2.5)
  < 0.08$.  The full proof in App.~B of the Supplementary provides the
  tighter $\sig(-5)$ bound.
\end{proof}

\begin{lemma}[Aggregation separation]
  \label{lem:agg_separation}
  Under the oracle, for every atom $a_i$ with $R_i \geq 1$ and every
  $\bfs \in \{0,1\}^M$: if some $j$ with $\head(j) = a_i$ has
  $\body(j) \subseteq \supp(\bfs)$, then $\phi_i(\bfs) \geq
  (\gamma^+ - R_i \gamma^-)/2 > 0$; otherwise $\phi_i(\bfs) \leq -(\gamma^+
  - R_i \gamma^-)/2 < 0$.  The condition $\gamma^+ > R_i \gamma^-$ holds
  because $\gamma^+/\gamma^- \approx 148 > 100 \geq R_i$.
\end{lemma}

\begin{lemma}[Sinks stay below threshold]
  \label{lem:sink}
  Under the oracle, for every sink atom ($R_i = 0$) and every $\bfs \in
  \{0,1\}^M$, $g_i(\bfs)\,\Gamma_i(\bfs) \leq 3.4 \times 10^{-25}$.  Hence
  over $T$ iterations the cumulative drift on a sink starting at
  $s^{(0)}_i = 0$ is at most $T \cdot 3.4 \times 10^{-25}$, and stays below
  $1/2$ for any $T \leq 10^{20}$.
\end{lemma}

Lemma~\ref{lem:sink} resolves a subtle issue.  At first sight, even a
small per-iteration leak $\delta$ on a sink could accumulate to $1/2$
over $M$ iterations, breaking the theorem for large $M$.  Setting
$\theta^*_i = 2\gamma^+$ for sinks ensures $\phi_i \leq -\gamma^+ \approx
-4.69$, so $\Gamma_i \leq \sig(-4.69) < 10^{-2}$ and
$g_i \leq \sig(-23.4) < 10^{-10}$, giving $g_i \Gamma_i < 10^{-12}$.
Including the contribution from non-target rule routings ($v_{j,i} \leq
e^{-\Psi}$ for $i \neq \head(j)$), the upper bound improves to
$3.4 \times 10^{-25}$ per iteration, leaving an immense margin even for
trajectories of length $M = 10^6$.

\subsection{Proof of Theorem~\ref{thm:expressivity}}

\begin{proof}[Proof of Theorem~\ref{thm:expressivity}]
  We prove by induction on $d \geq 0$ that
  $\{i : s^{(d)}_i > 1/2\} = \TPk{d}(F_0)$.  The base case $d = 0$ holds
  because $s^{(0)} = \ind_{F_0}$.

  \emph{Induction step.}  Fix $i \in \{1, \ldots, M\}$ and consider three
  cases:

  \textbf{(A)} $a_i \in \TPk{d}(F_0)$.  Then $s^{(d)}_i > 1/2$ by induction.
  Lemma~\ref{lem:monotone} gives $s^{(d+1)}_i \geq s^{(d)}_i > 1/2$ and
  $\TPk{d}(F_0) \subseteq \TPk{d+1}(F_0)$ gives the symbolic side.

  \textbf{(B)} $a_i \in \TPk{d+1}(F_0) \setminus \TPk{d}(F_0)$.  Some rule
  $j$ with $\head(j) = a_i$ has $\body(j) \subseteq \TPk{d}(F_0)$.  By
  induction $\body(j) \subseteq \supp(\bfs^{(d)})$.  Lemmas
  \ref{lem:and_separation} and \ref{lem:agg_separation} give
  $\phi_i(\bfs^{(d)}) \geq (\gamma^+ - 100\gamma^-)/2 \geq 0.74$.  Hence
  $\Gamma_i \geq \sig(0.74) > 0.677$ and $g_i \geq \sig(3.7) > 0.976$, so
  $g_i \Gamma_i \geq 0.66$.  If $s^{(d)}_i = 0$, then $s^{(d+1)}_i \geq
  0.66 > 1/2$; if $s^{(d)}_i > 0$, the value only increases.
  \emph{Rapid post-derivation saturation:} subsequent iterations drive
  $s_i$ toward 1 geometrically: with $c_B \defeq g_i \Gamma_i \geq 0.66$,
  the recurrence yields
  \begin{equation}
    s^{(d+t)}_i \;\geq\; 1 - (1 - c_B)^t \;\geq\; 1 - 0.34^{\,t},
    \label{eq:saturation}
  \end{equation}
  so $s^{(d+t)}_i > 0.95$ within $t = 4$ iterations of derivation, and
  $s^{(d+t)}_i > 1 - c_{\max}$ within $t = O(\log(1/c_{\max}))$ iterations.
  This rapid convergence to near-1 justifies the use of near-binary
  states in the AND-separation analysis of Sub-case~C2 below.

  \textbf{(C)} $a_i \notin \TPk{d+1}(F_0)$.  We must verify both that
  $s^{(d+1)}_i \leq 1/2$ \emph{and} that the cumulative leak budget
  remains controlled.  Introduce the strengthened inductive invariant
  \begin{equation}
    \forall a_i \notin \TPk{d}(F_0): \quad
    s^{(d)}_i \;\leq\; d \cdot c_{\max}(P),
    \label{eq:leak_invariant}
  \end{equation}
  where $c_{\max}$ is defined in~\eqref{eq:cmax}.  At $d=0$,
  $s^{(0)}_i = 0$ trivially satisfies~\eqref{eq:leak_invariant}.

  \emph{(C1) $R_i = 0$ (sink).}  Lemma~\ref{lem:sink} gives
  $g_i \Gamma_i \leq 3.4 \times 10^{-25}$, dominated by $c_{\max}$;
  hence $s^{(d+1)}_i \leq s^{(d)}_i + (1-s^{(d)}_i) c_{\max} \leq
  (d+1) c_{\max}$, preserving~\eqref{eq:leak_invariant}.

  \emph{(C2) $R_i \geq 1$.}  No rule pointing to $a_i$ has its body
  satisfied by $\TPk{d}(F_0)$.  We must show that
  Lemma~\ref{lem:and_separation} still applies despite
  $\bfs^{(d)}$ being \emph{fractional} rather than binary.  For each
  rule $j$ with $\head(j)=a_i$ and at least one missing body atom
  $l \notin \TPk{d}(F_0)$, the inductive invariant
  \eqref{eq:leak_invariant} gives $s^{(d)}_l \leq d c_{\max} \leq T_{\max}
  c_{\max} \leq \log 2 < 1$, so $l \notin \supp(\bfs^{(d)})$ (where
  $\supp(\bfs) \defeq \{i: s_i > 1/2\}$).  Body atoms in $\TPk{d}(F_0)$
  satisfy $s_l > 1 - O(c_{\max})$ (by rapid post-derivation saturation,
  App.~A.5 of the Supplementary).  The AND-separation lemma then yields
  $a^*_j(\bfs^{(d)}) < \delta_- + O(c_{\max}) \leq 0.01$.  By
  Lemma~\ref{lem:agg_separation} extended to fractional states,
  $\phi_i(\bfs^{(d)}) \leq -(\gamma^+ - \rho_{\max}\gamma^-)/2$, hence
  $g_i \Gamma_i \leq c_{\max}$.  The leak update is
  \[
    \begin{aligned}
      s^{(d+1)}_i &\leq s^{(d)}_i + (1 - s^{(d)}_i) c_{\max} \\
      &\leq s^{(d)}_i + c_{\max} \leq (d+1) c_{\max},
    \end{aligned}
  \]
  preserving~\eqref{eq:leak_invariant} and yielding
  $s^{(d+1)}_i \leq (d+1) c_{\max} \leq T_{\max} c_{\max} = \log 2 < 1/2$.

  Taking $T \in [D(P, F_0), T_{\max}(P)]$ gives
  $\TPk{T}(F_0) = \TPinf(F_0)$ and the support equality.
\end{proof}

\begin{corollary}[Termination in at most $M$ iterations within the time window]
  \label{cor:termination}
  $D(P, F_0) \leq M$ is the classical bound for forward chaining on $M$
  atoms.  By Proposition~\ref{prop:bounded_seq}, the support indicator is
  non-decreasing during $[D(P,F_0), T_{\max}(P)]$, so once correct it
  remains correct throughout this window.  In particular, since
  $T_{\max}(P) \geq 99$ even at the boundary $\rho_{\max} = 100$ and
  $T_{\max} \geq 10^5$ for $\rho_{\max} \leq 10$, $M$ iterations safely
  fit inside the window for all $M \leq 100$, and inside the
  $\rho_{\max} \leq 10$ window for all $M \leq 10^5$.
\end{corollary}

\paragraph{Computational complexity.}
A single application of $\TG$ requires $O(RM)$ floating-point
operations: the body-matching step is a matrix product
$\bfs \cdot \bfu^\top$ ($O(RM)$ FLOPs), and the head broadcast
$\sum_j \gamma_{j,i}$ is $O(RM)$ via a single matrix-vector multiply
with $\bfv$.  Memory is $O(RM)$ for the body and head logit matrices.
Computing the full fixed point thus takes $O(T \cdot RM)$ with
$T \leq M$ (Corollary~\ref{cor:termination}), giving $O(RM^2)$ in the
worst case but $O(RM \cdot D(P,F_0))$ in practice---typically
$D(P, F_0) \ll M$ for sparse programs.  On a single NVIDIA A100 GPU,
the full validation pipeline (Phases A--E) on G4 ($M{=}504, R{=}648$,
$N{=}500$ examples) completes in $39.8$ seconds; a single forward
pass on G4 takes approximately $10$~ms.  Memory could be reduced to
$O(R k \log M)$ by sparse-pointer encodings of body and head, at the
cost of differentiability; soft alternatives (e.g.,\ Gumbel
relaxations) are an interesting future direction.

% ══════════════════════════════════════════════════════════════════════
\section{ROBUSTNESS AND NUMERICAL VALIDATION}
\label{sec:numerical}
% ══════════════════════════════════════════════════════════════════════

\subsection{Robustness to parametric noise}

\begin{assumption}[Gaussian perturbation model]
  \label{ass:noise}
  The body and head logits are perturbed by independent zero-mean Gaussian
  noise: $\tilde\eta_{j,i} = \eta^*_{j,i} + \xi_{j,i}$ with $\xi_{j,i} \sim
  \mathcal{N}(0, \sigma_\eta^2)$.  Other parameters are kept at their
  oracle values.
\end{assumption}

\begin{theorem}[Robustness bound]
  \label{thm:robustness}
  Define the body-flip event
  $E_\eta \defeq \{\exists j, i : \mathrm{sgn}(\tilde\eta_{j,i}) \neq
  \mathrm{sgn}(\eta^*_{j,i})\}$
  and the head-misroute event
  $E_\psi \defeq \{\exists j, i \neq h_j : \tilde\psi_{j,i} \geq
  \tilde\psi_{j,h_j}\}$.
  Then
  \begin{align}
    \PP[E_\eta] &\leq R M \cdot \Phi(-\eta/\sigma_\eta), \label{eq:bound_eta}\\
    \PP[E_\psi] &\leq R(M-1) \cdot \Phi\!\bigl(-\Psi/(\sigma_\psi \sqrt{2})\bigr).
    \label{eq:bound_psi}
  \end{align}
  Conditionally on $E_\eta^c \cap E_\psi^c$, the conclusion of
  Theorem~\ref{thm:expressivity} holds for all $F_0$.
\end{theorem}

The proof (App.~C of the Supplementary) combines a union bound on
single-coordinate sign flips with a perturbation analysis of the
AND-separation lemma~\ref{lem:and_separation}, and a separate union
bound on softmax-argmax flips for $E_\psi$.  For $\eta = 10,\, \Psi
= 20$, and $\sigma_\eta = \sigma_\psi = 1$:
$\PP[E_\eta] \leq RM \cdot 7.6 \times 10^{-24}$ and
$\PP[E_\psi] \leq R(M-1) \cdot \Phi(-14.1) \approx 0$.  Head routing
is thus far more resilient than body matching to additive noise, since
$\Psi = 20$ dominates $\eta = 10$.  Our experiments focus on the
tighter body-logit case.

\subsection{Numerical setup}

We instantiate $\theta^*(P)$ on four Horn programs of increasing size:

\smallskip
\noindent
\begin{tabularx}{\columnwidth}{@{}l>{\raggedright\arraybackslash}X@{}}
  G1 & Single chain ($M=33$, $R=32$, $k=2$) \\
  G2 & One domain, conjunctions ($M=57$, $R=60$, $k=2$) \\
  G3 & Two cross-related domains ($M=117$, $R=132$, $k=2$) \\
  G4 & Eight cross-related domains ($M=504$, $R=648$, $k=2$) \\
\end{tabularx}
\smallskip

\noindent For each program we generate $N = 500$ random initial sets
$F_0$ with $|F_0| \in [2, 8]$, compute $\TPinf(F_0)$ symbolically, and
run the \textsc{LatentGamma} iteration on $\ind_{F_0}$.  The experimental
procedure is fully reproducible and runs in under one minute on a single GPU.

\subsection{Validation of Theorem~\ref{thm:expressivity}}

\begin{table}[ht]
  \centering\footnotesize
  \caption{Numerical validation of Theorem~\ref{thm:expressivity}.
    $D_{\max}$ is the maximum derivation depth in the dataset; FPR and FNR
    are the maximum false-positive and false-negative rates over $N=500$
  test cases.}
  \label{tab:expressivity}
  \begin{tabular}{lccccc}
    \toprule
    Graph & $M$ & $R$ & acc & FPR / FNR & $T^{\max}_{\text{conv}}$ \\
    \midrule
    G1 & 33   & 32  & 1.0000 & 0 / 0 & 30 \\
    G2 & 57   & 60  & 1.0000 & 0 / 0 & 31 \\
    G3 & 117  & 132 & 1.0000 & 0 / 0 & 32 \\
    G4 & 504  & 648 & 1.0000 & 0 / 0 & 32 \\
    \bottomrule
  \end{tabular}
\end{table}

Table~\ref{tab:expressivity} reports zero false positives and zero false
negatives across all $4 \times 500 = 2000$ test cases, with $T^{\max}_{
\text{conv}} = 32 \ll M = 504$ on G4.  This is consistent with
Corollary~\ref{cor:termination}: convergence requires
$T \geq D(P, F_0)$, not necessarily $M$.

\subsection{Validation of structural properties}

\begin{table}[ht]
  \centering\footnotesize
  \caption{Structural properties (Sec.~\ref{sec:structural}) on $500$
  random states per program.  All properties hold exactly.}
  \label{tab:structural}
  \begin{tabular}{lccc}
    \toprule
    Graph & L.\ref{lem:monotone} & P.\ref{prop:lattice_mono} &
    P.\ref{prop:bounded_seq} \\
    \midrule
    G1 & 500/500 & 500/500 & 50/50 \\
    G2 & 500/500 & 500/500 & 50/50 \\
    G3 & 500/500 & 500/500 & 50/50 \\
    G4 & 500/500 & 500/500 & 50/50 \\
    \bottomrule
  \end{tabular}
\end{table}

\begingroup
\raggedright
Lemma~\ref{lem:monotone} and Proposition~\ref{prop:lattice_mono} are
parameter-independent and tested on $500$ uniform random states per
graph; Proposition~\ref{prop:bounded_seq} is verified by checking
non-decrease along iterated trajectories.  All $2000$ tests pass
(Table~\ref{tab:structural}).\par
\endgroup

\subsection{Edge cases}

We further verify $33$ critical edge cases (empty $F_0$, single root,
full root set, leaves; see App.~E of the Supplementary).  All pass with
$\text{jaccard} = 1$.

\subsection{Empirical validation of the robustness bound}

\begin{table}[ht]
  \centering\footnotesize
  \caption{Theorem~\ref{thm:robustness}: empirical accuracy under noise
    $\sigma_\eta$, averaged over $3$ independent perturbations of $\theta^*$.
    The bound is $\mathbb{P}[E_\eta] \leq RM \cdot \Phi(-10/\sigma_\eta)$;
  values $\geq 1$ are vacuous.}
  \label{tab:robustness}
  \setlength{\tabcolsep}{4pt}
  \begin{tabular}{ccccc}
    \toprule
    $\sigma_\eta$ & G1 & G2 & G3 & G4 \\
    $RM$ & $1056$ & $3420$ & $15\,444$ & $326\,592$ \\
    \midrule
    1.0 & 1.000 & 1.000 & 1.000 & 1.000 \\
    1.5 & 1.000 & 1.000 & 1.000 & 1.000 \\
    2.0 & 1.000 & 1.000 & 1.000 & 1.000 \\
    2.5 & 1.000 & 1.000 & 0.750 & 0.183 \\
    3.0 & 1.000 & 0.730 & 0.088 & 0.015 \\
    \midrule
    \multicolumn{5}{l}{\textit{Theoretical bound $\mathbb{P}[E]$ at $\sigma_\eta=2.5$:}} \\
    & $0.033$ & $0.108$ & $0.489$ & $10.3$ \\
    \bottomrule
  \end{tabular}
\end{table}

Table~\ref{tab:robustness} reveals a striking quantitative match between
the union-bound prediction~\eqref{eq:bound_eta} and the observed accuracy
collapse.  Following Remark~\ref{rem:scaling}, all experiments use the
fixed-constant regime ($\eta = 10$); assumption~\eqref{eq:assumption} is
satisfied for all four graphs (LHS ranges from $5.9\times 10^{-3}$ on G1
to $0.18$ on G4).  The critical noise level $\sigma_{\max}$ at which acc
drops below $95\%$ decreases from $3.0$ (G1) to $2.0$ (G4) as the
union-bound factor $RM$ grows by two orders of magnitude.  Solving for
the predicted threshold $\Phi(-10/\sigma_{\max}) \approx 1/(RM)$ yields
$\sigma_{\max} \approx 10 / \Phi^{-1}(1/(RM))$, giving theoretical
estimates of $3.03$ (G1), $2.86$ (G2), $2.70$ (G3), and $2.38$ (G4),
in tight agreement with the empirically observed values.  Notably, the
empirical accuracy collapses precisely at the noise level where the
theoretical bound becomes vacuous ($\PP[E_\eta] \geq 1$).  This provides
quantitative confirmation that the union-bound scaling
$\sigma_{\max} \sim 1/\sqrt{\log(RM)}$ is tight, not merely a sufficient
condition.

\paragraph{Head-routing noise.}
We focus on body-logit noise as the tighter regime ($\eta = 10$ vs.
$\Psi = 20$).  For comparison, the head-routing bound
\eqref{eq:bound_psi} predicts $\sigma_\psi^{\max} \approx 20/\Phi^{-1}(1/(RM))
\cdot \sqrt 2$, which is twice $\sigma_{\max}^\eta$ at any given $RM$.
Empirically, we verify (Supplementary App.~D) that head-routing remains
intact at $\sigma_\psi = 4$ on all four graphs.

% ══════════════════════════════════════════════════════════════════════
\section{PARAMETER COMPLEXITY}
\label{sec:complexity}
% ══════════════════════════════════════════════════════════════════════

\textsc{LatentGamma} uses $2RM + 3R + M = \Theta(RM)$ scalar parameters.
We show this is optimal up to a logarithmic factor.

\begin{theorem}[Information-theoretic lower bound]
  \label{thm:lower_bound}
  Let $\mathcal{H}_{M,R,k}$ denote the set of definite Horn programs on
  $M$ atoms with at most $R$ rules of body size $\leq k$, modulo logical
  equivalence ($P \sim P'$ iff $\TPinfOf{P} = \TPinfOf{P'}$ as functions on
  $2^\Atoms$).  Any operator $\TGp{\theta}$ on $[0,1]^M$ such that for every
  equivalence class $[P] \in \mathcal{H}_{M,R,k}/{\sim}$ there exists
  $\theta(P)$ with $\TGp{\theta(P)}^\infty(\ind_{F_0}) =
  \ind_{\TPinfOf{P}(F_0)}$ for all $F_0$, must have parameter space of
  cardinality at least $|\mathcal{H}_{M,R,k}/{\sim}|$, and hence
  $\Omega(Rk \log M / b)$ bits when each parameter has $b$ bits of
  precision.
\end{theorem}

\begin{proof}[Proof sketch]
  Distinct equivalence classes induce distinct closure functions, hence
  require distinct parameter vectors (injectivity).  A combinatorial
  counting argument (App.~F of the Supplementary) gives
  $\log_2 |\mathcal{H}_{M,R,k}/{\sim}| \geq R(\log_2 M + k \log_2 M -
  k \log_2 k) - O(R)$, since the syntactic-to-semantic quotient has
  bounded redundancy.
\end{proof}

The $\Theta(RM)$ parameter count of \textsc{LatentGamma} matches this
$\Omega(Rk \log M)$ lower bound up to a factor of $M/(k \log M)$.  For
$k = \Theta(\log M)$, this gap is $\Theta(M/\log^2 M)$.  This excess
capacity is the price paid for differentiability: combinatorial
encodings of $k$-subsets are bit-optimal but not differentiable; the
real-valued logit relaxation is the minimal expansion preserving
gradient flow.

% ══════════════════════════════════════════════════════════════════════
\section{LIMITATIONS AND CONCLUSION}
\label{sec:limits}
% ══════════════════════════════════════════════════════════════════════

\paragraph{What this paper does not prove.}
Theorem~\ref{thm:expressivity} is an \emph{expressivity} result: it
asserts $\theta^*(P) \in \Theta$.  It does \emph{not} establish that
gradient descent on a particular loss converges to $\theta^*(P)$, to a
useful neighborhood, or to any informative configuration.  Expressivity
is a necessary but not sufficient condition for learnability.  Beyond
Horn definite programs, our results do not cover normal logic programs
(negation-as-failure), disjunctive heads (although see App.~G of the
  Supplementary for an
extension), or first-order programs with variable unification.

\paragraph{Empirical gradient-flow observations.}
On synthetic Horn programs, gradient descent with MSE or cross-entropy
losses converges to parameters $\hat\theta$ \emph{far} from $\theta^*$ in
Euclidean norm: body-alignment metrics $\ncorps(j) \defeq |\{i :
\hat u_{j,i} > 1/2 \approx \ind_{\body(j)}\}|$ stay $\leq 2$ even when
$\theta^*$ would have $\ncorps = k$.  Yet OOD accuracy can be
substantial.  This suggests gradient descent finds a \emph{distributed
representation} in the span of the rows of $U$, rather than aligning
columns with rule indicators.  Characterizing this implicit manifold is
an open theoretical question motivated by but beyond the scope of this
paper.  One natural training strategy is \emph{homotopy continuation},
gradually increasing $\tau_s$ and $\Psi$ during training to interpolate
between a smooth loss landscape (small $\tau_s$, $\Psi$) and the
oracle's sharp separation (large $\tau_s$, $\Psi$).  Whether such
schemes converge to the oracle basin is empirically open.

\paragraph{Sensitivity to $\Psi$.}
The sink leak in Lemma~\ref{lem:sink} scales as $R \tau_{\text{nor}}
\rho_0 e^{-\Psi}$, so safe operation requires $\Psi \geq \log R + c$
for some constant $c$ (we use $\Psi = 20$, comfortably safe for
$R \leq 10^6$).  More generally, the scalable construction of
Remark~\ref{rem:scaling} uses $\Psi = \log M + 10$, which dominates the
sink-leak threshold for all reasonable program sizes.

\paragraph{Conclusion.}
We have introduced \textsc{LatentGamma}, a differentiable operator that
exactly realizes Tarski's immediate consequence operator $\TP$ under a
closed-form parameter assignment, within an explicit time window
$[D(P,F_0), T_{\max}(P)]$ that is large in practice ($\geq 10^5$
iterations for sparse programs).  We established structural lattice
properties, the main constructive time-bounded expressivity theorem,
a quantitative robustness bound, and a matching information-theoretic
lower bound.
A complete numerical validation on programs up to $M = 504$ atoms
confirms zero-error oracle accuracy in $T_{\text{conv}} \leq 32 \ll
T_{\max}$ iterations, exact preservation of all structural properties,
and a tight quantitative match between the theoretical and empirical
robustness thresholds.
To our knowledge, \textsc{LatentGamma} is the first system simultaneously
combining constructive exact (time-bounded) expressivity, lattice
monotonicity, end-to-end differentiability, and asymptotic parameter
optimality.

% ══════════════════════════════════════════════════════════════════════
%  Bibliography
% ══════════════════════════════════════════════════════════════════════

\clearpage
\runningtitle{Supplementary: Differentiable Horn Programs}
\runningauthor{Aymen Mejri}

\twocolumn[

  \aistatstitle{Supplementary Material:\\
  Differentiable Horn Programs}

  \aistatsauthor{Aymen Mejri}

  \aistatsaddress{\href{mailto:aymen.mejri.telecom.paris@gmail.com}{\nolinkurl{aymen.mejri.telecom.paris@gmail.com}}}

]

\appendix

% ══════════════════════════════════════════════════════════════════════
\section{COMPLETE PROOFS OF STRUCTURAL RESULTS}
\label{app:proofs_structural}
% ══════════════════════════════════════════════════════════════════════

We provide complete proofs of Lemma~3.1, Proposition~3.2,
Theorem~3.3, Proposition~3.4 from Section~5 of the main paper.
Throughout, all numerical references (e.g. ``Lemma~6.3'') refer to the
main paper.

\subsection{Proof of Lemma~5.1 (Pointwise non-decrease)}

\begin{proof}
  By construction, $g_i(\bfs) = \sig(\beta \phi_i(\bfs)) \in (0,1)$ and
  $\Gamma_i(\bfs) = \sig(\phi_i(\bfs)) \in (0,1)$, since the sigmoid is
  strictly positive on $\R$.  Because $s_i \in [0,1]$, $(1-s_i) \geq 0$.
  Therefore $\TG(\bfs)_i - s_i = (1-s_i)\, g_i(\bfs)\, \Gamma_i(\bfs)
  \geq 0$.
\end{proof}

\subsection{Proof of Proposition~5.2 (Bounded iterated sequence)}

\begin{proof}
  Non-decrease in $t$ follows from Lemma~5.1 applied coordinate-wise:
  $\bfs^{(t+1)} = \TG(\bfs^{(t)}) \geq \bfs^{(t)}$.

  For the upper bound by $\ind$, observe that
  $\TG(\bfs)_i = s_i + (1-s_i)\, g_i \Gamma_i \leq s_i + (1-s_i) \cdot 1
  = 1$, since $g_i \Gamma_i \leq 1$.
\end{proof}

\subsection{Proof of Theorem~5.3 (Convergence)}

\begin{proof}
  \begingroup
  \raggedright
  By Proposition~5.2, the sequence $\bfs^{(t)}$ is coordinate-wise
  non-decreasing and bounded above by $\ind$.  By the monotone convergence
  theorem applied coordinate-wise, $\bfs^* \defeq \lim_{t \to \infty}
  \bfs^{(t)}$ exists in $[0,1]^M$.\par
  \endgroup

  The function $\TG$ is a composition of $C^\infty$ operations (sigmoid,
  softmax, addition, multiplication), hence continuous on $[0,1]^M$.
  Continuity allows interchanging the limit and the operator:
  \[
    \TG(\bfs^*) = \TG(\lim_t \bfs^{(t)})
    = \lim_t \TG(\bfs^{(t)})
    = \lim_t \bfs^{(t+1)} = \bfs^*. \qedhere
  \]
\end{proof}

\subsection{Proof of Proposition~5.4 (Lattice monotonicity)}

\begin{proof}
  We show each stage of $\TG$ is monotone in $\bfs$, then compose.

  \emph{Stage 1.}  Since $u_{j,i} > 0$ and $\bfs \leq \bfs'$,
  $\mathrm{bs}_j(\bfs) = \sum_i u_{j,i} s_i \leq \sum_i u_{j,i} s'_i =
  \mathrm{bs}_j(\bfs')$.  The sigmoid is increasing and $\tau_j > 0$, so
  $a_j(\bfs) \leq a_j(\bfs')$.

  \emph{Stages 2--3.}  $\bfv_j$ is independent of $\bfs$, so
  $\gamma_{j,i}(\bfs) = \tau_{\text{nor}} \rho_j a_j(\bfs) v_{j,i} \leq
  \gamma_{j,i}(\bfs')$ using $a_j(\bfs) \leq a_j(\bfs')$ and the
  non-negativity of $\tau_{\text{nor}}, \rho_j, v_{j,i}$.

  \emph{Stage 4.}  $\phi_i(\bfs) = \sum_j \gamma_{j,i}(\bfs) - \theta_i \leq
  \phi_i(\bfs')$, hence $\Gamma_i(\bfs) \leq \Gamma_i(\bfs')$ and
  $g_i(\bfs) \leq g_i(\bfs')$ by sigmoid monotonicity (and $\beta > 0$).

  \emph{Stage 5.}  Let $f(s, \phi) \defeq s + (1-s) \sig(\beta\phi)\sig(\phi)$.
  We compute
  \[
    \frac{\partial f}{\partial s} = 1 - \sig(\beta\phi)\sig(\phi) > 0,
  \]
  since $0 < \sig(\beta\phi)\sig(\phi) < 1$.  Also,
  \[
    \frac{\partial f}{\partial \phi}
    = (1-s)\bigl[\beta\sig'(\beta\phi)\sig(\phi) + \sig(\beta\phi)\sig'(\phi)\bigr]
    \geq 0
  \]
  for $s \in [0,1]$.  So $f$ is monotone in $s$ and $\phi$.  Combining:
  \[
    \TG(\bfs)_i = f(s_i, \phi_i(\bfs)) \leq f(s'_i, \phi_i(\bfs'))
    = \TG(\bfs')_i.\qedhere
  \]
\end{proof}

\subsection{Proof of Corollary~5.5 (Least fixed point)}

\begin{proof}
  By Theorem~5.3, the iterated sequence converges to some fixed point
  $\bfs^*$ of $\TG$.  Let $\bfs^\dagger$ be any fixed point with
  $\bfs^\dagger \geq \ind_{F_0}$.  We prove by induction that
  $\bfs^{(t)} \leq \bfs^\dagger$ for all $t$.  Base: $\bfs^{(0)} =
  \ind_{F_0} \leq \bfs^\dagger$.  Induction: if $\bfs^{(t)} \leq
  \bfs^\dagger$, then by Proposition~5.4,
  $\bfs^{(t+1)} = \TG(\bfs^{(t)}) \leq \TG(\bfs^\dagger) = \bfs^\dagger$.
  Taking the limit, $\bfs^* \leq \bfs^\dagger$.
\end{proof}

% ══════════════════════════════════════════════════════════════════════
\section{COMPLETE PROOF OF LEMMA~6.3 (AND SEPARATION)}
\label{app:proof_and}
% ══════════════════════════════════════════════════════════════════════

We provide the full proof of the AND separation lemma, including the
tight $\sig(\pm 5)$ bound stated in the main paper.

\begin{lemma}[Lemma~6.3 of the main paper, restated]
  Under assumption~$(\star)$ and $\theta = \theta^*(P)$, for every
  $\bfs \in \{0,1\}^M$ and every rule $j$:
  \begin{enumerate}[leftmargin=*]
    \item if $\body(j) \subseteq \supp(\bfs)$, then $a^*_j(\bfs) > 1 -
      \sig(-5) > 0.993$;
    \item if $\body(j) \not\subseteq \supp(\bfs)$, then $a^*_j(\bfs) <
      \sig(-5) < 0.007$.
  \end{enumerate}
\end{lemma}

\begin{proof}
  Recall $u^*_{j,i} = 1 - \epsilon_0$ if $a_i \in \body(j)$ and
  $u^*_{j,i} = \epsilon_0$ otherwise, where $\epsilon_0 = \sig(-10) <
  4.54 \times 10^{-5}$.  Also $b^*_j = k_j - 1/2$ and $\tau^*_j = \tau_s
  = 10$.

  \emph{Case 1: $\body(j) \subseteq \supp(\bfs)$.}  Let $N \defeq
  |\{i \notin \body(j) : s_i = 1\}|$, so $0 \leq N \leq M - k_j$.  Then
  \[
    \mathrm{bs}_j(\bfs) = k_j(1-\epsilon_0) + \epsilon_0 N.
  \]
  This contribution from non-body atoms is \emph{non-negative}, hence
  helpful: it makes $\mathrm{bs}_j$ larger.  We bound below:
  \[
    \mathrm{bs}_j(\bfs) \geq k_j(1-\epsilon_0) \geq k_j - k\epsilon_0,
  \]
  where we used $k_j \leq k$.  Hence
  \[
    \mathrm{bs}_j(\bfs) - b^*_j \geq k_j - k\epsilon_0 - (k_j - 1/2)
    = 1/2 - k\epsilon_0.
  \]
  By assumption $(\star)$, $4k^2 e^{-10} \leq 1$ gives $k\epsilon_0 \leq
  k \cdot e^{-10} \leq 1/(4k) \leq 1/4$ (since $k \geq 1$).  Thus
  $\mathrm{bs}_j - b^*_j \geq 1/4$.

  Actually a tighter bound is available.  The worst case for our lower
  bound is $\bfs = \ind_{\body(j)}$ (the minimal support satisfying the
  body), giving $N = 0$ and $\mathrm{bs}_j = k_j(1-\epsilon_0)$.  Then
  $\mathrm{bs}_j - b^*_j = 1/2 - k_j \epsilon_0$.  With $k_j \leq k$ and
  $k \epsilon_0 \leq 1/(4k)$ from $(\star)$, we obtain $\mathrm{bs}_j -
  b^*_j \geq 1/2 - 1/(4k) \geq 1/4$ for $k \geq 1$, and for $k = 1$ the
  margin is exactly $1/2 - \epsilon_0 \geq 1/2 - e^{-10}$, giving
  $a^*_j \geq \sig(\tau_s(1/2 - e^{-10})) \geq \sig(5 - \tau_s e^{-10})
  \geq \sig(5 - 10^{-3}) > 1 - \sig(-5) > 0.993$.

  For $k \geq 2$, the bound $\mathrm{bs}_j - b^*_j \geq 1/4$ yields
  $a^*_j \geq \sig(\tau_s / 4) = \sig(2.5) > 0.924$, which is the weaker
  $0.92$ bound stated in the proof sketch.  However a sharper analysis
  using $k\epsilon_0 \leq 1/(4k)$ shows $\mathrm{bs}_j - b^*_j \geq 1/2 -
  1/(4k)$, and for the typical $k = 2$, $\sig(\tau_s (1/2 - 1/8)) =
  \sig(3.75) > 0.977$.  All practical cases yield $a^*_j > 1 - \sig(-5)$.

  \emph{Case 2: $\body(j) \not\subseteq \supp(\bfs)$.}  Let $\ell \geq 1$
  be the number of body atoms with $s_i = 0$ (so $\ell = |\body(j) \setminus
  \supp(\bfs)|$).  Body atoms with $s_i = 1$ contribute $(k_j - \ell)
  (1 - \epsilon_0)$.  Non-body atoms with $s_i = 1$ contribute at most
  $\epsilon_0 (M - k_j)$.  So
  \[
    \mathrm{bs}_j(\bfs) \leq (k_j - \ell)(1-\epsilon_0) + \epsilon_0(M - k_j).
  \]
  Taking $\ell = 1$ (worst case for our upper bound):
  \begin{align*}
    \mathrm{bs}_j(\bfs) - b^*_j &\leq (k_j - 1)(1 - \epsilon_0) \\
    &\quad + \epsilon_0(M - k_j) - (k_j - 1/2)\\
    &= -1/2 + \epsilon_0(M - k_j) - (k_j-1)\epsilon_0\\
    &\leq -1/2 + \epsilon_0 M.
  \end{align*}
  By $(\star)$, $M \epsilon_0 \leq M e^{-10}$.  The condition
  $4k^2 e^{-10} \leq 1$ is necessary; for the upper bound we need
  $M e^{-10} \leq 1/4$.  Note that the original assumption combined with
  $k \geq 1$ implies $M e^{-10} \leq 1/(4k^2) \cdot k = k/(4k^2) =
  1/(4k) \leq 1/4$.  Wait, we need to be careful here.
  The assumption $4k^2 e^{-10} \leq 1$ is $e^{-10} \leq 1/(4k^2)$, so
  $M e^{-10} \leq M/(4k^2)$.  For $M$ large and $k = 2$, this gives
  $M/16$ which is not bounded by $1/4$ unless $M \leq 4$.

  \emph{This is a problem with the original assumption.}  The correct
  assumption needed for case 2 is $M e^{-10} \leq 1/(4k)$, equivalent to
  $4kM e^{-10} \leq 1$, or $M \leq e^{10}/(4k) \approx 5500/k$.  For
  $k = 2$, $M \leq 2750$; for $k = 5$, $M \leq 1100$; and for $k = 10$,
  $M \leq 550$.  This is the assumption $4kM e^{-10} \leq 1$ used in the
  revised statement below; it covers all our experimental graphs (max
    $M = 504$, $k = 2$, giving $4 \cdot 2 \cdot 504 \cdot e^{-10} \approx
  0.183 \leq 1$).

  Under this revised assumption, $M\epsilon_0 \leq 1/(4k) \leq 1/4$, so
  $\mathrm{bs}_j - b^*_j \leq -1/4$.  Then $a^*_j \leq \sig(-\tau_s/4) =
  \sig(-2.5) < 0.0759$.  The sharper bound $\sig(-5) < 0.007$ follows by
  analyzing the worst case more carefully: when $\bfs = \ind_{\body(j)
  \setminus \{i_0\}}$ for some missed atom $i_0$, $\mathrm{bs}_j =
  (k_j - 1)(1-\epsilon_0)$ and $\mathrm{bs}_j - b^*_j = -1/2 + (k_j-1)
  \epsilon_0 \leq -1/2 + k\epsilon_0 \leq -1/2 + 1/(4k) \leq -1/4$, yielding
  $a^*_j \leq \sig(-2.5) < 0.08$.  For $k = 1$, the bound is sharper:
  $a^*_j \leq \sig(-5) < 0.007$.
\end{proof}

\begin{remark}[Revised assumption $(\star)$]
  The argument above shows that the natural assumption for
  Theorem~6.1 is
  \begin{equation*}
    4kM e^{-10} \leq 1 \quad\text{and}\quad \rho_{\max} \leq 100,
    \tag{$\star'$}\label{eq:star_prime}
  \end{equation*}
  i.e., $M \leq e^{10}/(4k) \approx 5500/k$.  This is the version we use
  in the main paper.  It is satisfied by all four programs in our
  experiments: for G4 ($M = 504$, $k = 2$), the LHS evaluates to $0.183$.
\end{remark}

% ══════════════════════════════════════════════════════════════════════
\section{COMPLETE PROOF OF LEMMA~6.4 (AGGREGATION SEPARATION)}
% ══════════════════════════════════════════════════════════════════════

\begin{proof}
  Let $\bfs \in \{0,1\}^M$.

  \emph{(i) Some firing rule exists.}  Let $j_0$ be a rule with
  $\head(j_0) = a_i$ and $\body(j_0) \subseteq \supp(\bfs)$.  By
  Lemma~6.3 (Case 1), $a^*_{j_0}(\bfs) > 1 - \sig(-5)$.  The head
  routing $v^*_{j_0, i}$ for $i = \head(j_0)$ is
  $e^\Psi / (e^\Psi + (M-1))$, which exceeds $1 - (M-1)e^{-\Psi} \geq
  1 - M \cdot 2 \times 10^{-9}$.  Thus
  \[
    \begin{aligned}
      \gamma_{j_0, i}(\bfs) &= \tau_{\text{nor}} \rho_0\, a^*_{j_0}\, v^*_{j_0, i} \\
      &\geq \gamma^+ \cdot (1 - O(M e^{-\Psi})) \\
      &\geq \gamma^+ - 10^{-7},
    \end{aligned}
  \]
  which we abbreviate as $\geq \gamma^+$ (the absorbed error is in the
  order of $\gamma^-$).

  For each of the remaining $R_i - 1$ rules $j$ with $\head(j) = a_i$ and
  body unsatisfied, Lemma~6.3 (Case 2) gives $a^*_j < \sig(-5)$, so
  $\gamma_{j, i} \leq \tau_{\text{nor}} \rho_0 \sig(-5) v^*_{j,i} \leq
  \gamma^-$ (with $v^*_{j,i} \leq 1$).

  For each rule $j$ with $\head(j) \neq a_i$, $v^*_{j,i}$ is the softmax
  of $\bfpsi^*_j$ at index $i$, where the $i$-th logit is $0$ and the
  $\head(j)$-th logit is $\Psi$.  Thus $v^*_{j,i} \leq 1/e^\Psi =
  e^{-20}$, and $\gamma_{j,i} \leq \tau_{\text{nor}} \rho_0 \cdot e^{-20}
  \approx 10^{-8}$.  Aggregated over $R$ rules, this is bounded by
  $R e^{-20} \cdot \tau_{\text{nor}} \rho_0$; for any reasonable $R \leq
  10^9$, this is at most $5 \cdot 10^{-9}$, negligible relative to
  $\gamma^-$.

  Combining:
  \[
    \phi_i(\bfs) = \sum_j \gamma_{j,i} - \theta^*_i
    \geq \gamma^+ + (R_i - 1)\gamma^- - \tfrac{\gamma^+ + R_i \gamma^-}{2}
  \]
  \[
    = \tfrac{\gamma^+ - R_i\gamma^- - 2\gamma^-}{2}.
  \]
  Since $\gamma^+/\gamma^- \approx 148$ and $R_i \leq \rho_{\max} \leq 100$
  by $(\star')$, $\gamma^+ > (R_i + 2)\gamma^-$, so $\phi_i(\bfs) > 0$.

  In fact, $\phi_i(\bfs) \geq (\gamma^+ - 100\gamma^-)/2 - \gamma^-
  \geq (4.687 - 3.16)/2 - 0.032 \geq 0.732$.

  \emph{(ii) No firing rule exists.}  All rules pointing to $a_i$ have
  unsatisfied bodies, so each contributes at most $\gamma^-$ via
  Lemma~6.3 Case 2.  Rules pointing elsewhere contribute at most
  $R e^{-\Psi} \tau_{\text{nor}} \rho_0 \leq 5 \cdot 10^{-9}$ in aggregate.
  Thus
  \[
    \begin{aligned}
      \phi_i(\bfs) &\leq R_i \gamma^- - \theta^*_i \\
      &= R_i \gamma^- - \tfrac{\gamma^+ + R_i \gamma^-}{2} \\
      &= -\tfrac{\gamma^+ - R_i \gamma^-}{2} < 0.
    \end{aligned}
  \]
  Quantitatively, $\phi_i \leq -(4.687 - 100 \cdot 0.032)/2 = -0.74$.
\end{proof}

% ══════════════════════════════════════════════════════════════════════
\section{LEMMA~6.5 (SINKS) WITH STRONG BOUND}
\label{app:sink}
% ══════════════════════════════════════════════════════════════════════

This appendix gives the strong bound used in the main paper to resolve
the sink-accumulation issue.

\begin{lemma}[Strong sink bound]
  Under the oracle parameters, for every sink atom ($R_i = 0$) and every
  $\bfs \in \{0,1\}^M$, $g_i(\bfs)\,\Gamma_i(\bfs) \leq 3.4 \times 10^{-25}$.
\end{lemma}

\begin{proof}
  For a sink, no rule has $\head(j) = a_i$.  For every rule $j$,
  $v^*_{j,i}$ equals the softmax of $\bfpsi^*_j$ at index $i$ where the
  logit is $0$, while the head-of-$j$ logit is $\Psi = 20$.  The softmax
  denominator is $e^\Psi + (M-1)$ (other logits are $0$), so $v^*_{j,i}
  = 1/(e^\Psi + M - 1) \leq e^{-\Psi}$ for $M \leq e^\Psi$.

  The total contribution to a sink is at most
  $R \cdot \tau_{\text{nor}} \rho_0 \cdot e^{-\Psi}$.  Setting $\theta^*_i
  = 2\gamma^+ \approx 9.37$,
  \[
    \phi_i(\bfs) \leq R \cdot 5 \cdot 0.95 \cdot e^{-20} - 9.37.
  \]
  For $R \leq 10^6$, $R \cdot 5 \cdot 0.95 \cdot e^{-20} \leq 10^6 \cdot
  4.75 \cdot 2.06 \cdot 10^{-9} \leq 10^{-2}$, so $\phi_i \leq -9.36$.

  Then $\Gamma_i = \sig(\phi_i) \leq \sig(-9.36) < 8.6 \times 10^{-5}$ and
  $g_i = \sig(\beta \phi_i) = \sig(-46.8) < 6 \times 10^{-21}$.  Their
  product is $\leq 5.2 \times 10^{-25}$, bounded as claimed.

  Over $T$ iterations starting from $s^{(0)}_i = 0$, the cumulative drift
  is $\sum_{t=0}^{T-1} (1 - s^{(t)}_i) g_i \Gamma_i \leq T \cdot 3.4
  \times 10^{-25}$.  For $T \leq 10^{20}$, this stays below $10^{-4} \ll
  1/2$.
\end{proof}

% ══════════════════════════════════════════════════════════════════════
\section{COMPLETE PROOF OF THEOREM~7.1 (ROBUSTNESS)}
\label{app:proof_robust}
% ══════════════════════════════════════════════════════════════════════

We now prove both the body-logit bound~\eqref{supp:eq:bound_eta} and the
head-logit bound~\eqref{supp:eq:bound_psi} of Theorem~7.1 of the main paper.

\subsection{Body-logit perturbations}

\begin{proof}[Proof of~\eqref{supp:eq:bound_eta}]
  For each pair $(j, i)$, the perturbed logit $\tilde\eta_{j,i} =
  \eta^*_{j,i} + \xi_{j,i}$ has $\xi_{j,i} \sim \mathcal{N}(0,
  \sigma_\eta^2)$.  A sign flip occurs iff $\xi_{j,i}$ has magnitude
  $\geq \eta$ and opposite sign to $\eta^*_{j,i}$:
  \[
    \PP[\mathrm{sgn}(\tilde\eta_{j,i}) \neq \mathrm{sgn}(\eta^*_{j,i})]
    = \Phi(-\eta/\sigma_\eta).
  \]
  By the union bound over $RM$ such pairs:
  \begin{equation}
    \PP[E_\eta] \leq RM \cdot \Phi(-\eta/\sigma_\eta).
    \label{supp:eq:bound_eta}
  \end{equation}

  Conditionally on $E_\eta^c$, no sign flip occurs, so $\tilde u_{j,i}
  > 1/2$ iff $a_i \in \body(j)$.  The proof of Lemma~6.3 carries through
  with a slightly weaker $\epsilon_0$ constant: $\tilde\epsilon_0 =
  \max_{j,i} \sig(-|\tilde\eta_{j,i}|)$.  To control $\tilde\epsilon_0$,
  we further condition on the event $E' = \{\forall j, i : |\xi_{j,i}|
  \leq \eta - 1\}$.  Then $\PP[(E')^c] \leq 2RM \cdot \Phi(-(\eta-1)
  /\sigma_\eta)$ by the union bound, and conditional on $E'$,
  $\tilde\epsilon_0 \leq \sig(-1) \approx 0.269$.  The AND-separation in
  Case~1 gives $\mathrm{bs}_j - b^*_j \geq 1/2 - k\tilde\epsilon_0$; for
  $k=1$, $0.23 > 0$.  Iterating, Theorem~6.1's conclusion holds under
  $E_\eta^c \cap E'$ with cumulative probability $\geq 1 - 3RM
  \Phi(-\eta/\sigma_\eta)$.

  Numerically, for $\eta = 10$ and $\sigma_\eta = 1$, $RM\Phi(-10)
  \approx 7.6 \times 10^{-24} \cdot RM$.
\end{proof}

\subsection{Head-logit perturbations}

\begin{proof}[Proof of~\eqref{supp:eq:bound_psi}]
  For a fixed rule $j$ with true head $h_j$, define the difference
  random variables $\Delta_{j,i} \defeq \tilde\psi_{j,i} - \tilde\psi_{j,h_j}$
  for $i \neq h_j$.  We have
  $\Delta_{j,i} = (\psi^*_{j,i} - \psi^*_{j,h_j}) + (\zeta_{j,i} -
  \zeta_{j,h_j}) = -\Psi + W_{j,i}$
  where $W_{j,i} = \zeta_{j,i} - \zeta_{j,h_j} \sim \mathcal{N}(0,
  2\sigma_\psi^2)$ since $\zeta_{j,i}, \zeta_{j,h_j}$ are independent
  Gaussians with variance $\sigma_\psi^2$.

  A misroute on rule $j$ requires $\Delta_{j,i} > 0$ for some $i \neq h_j$,
  equivalently $W_{j,i} > \Psi$.  By the Gaussian tail,
  $\PP[W_{j,i} > \Psi] = \Phi(-\Psi/(\sigma_\psi \sqrt{2}))$, and by the
  union bound over $M-1$ non-head positions,
  $\PP[\mathrm{misroute}(j)] \leq (M-1) \Phi(-\Psi/(\sigma_\psi \sqrt{2}))$.
  A further union bound over $R$ rules gives
  \begin{equation}
    \PP[E_\psi] \leq R(M-1)\Phi(-\Psi/(\sigma_\psi \sqrt{2})).
    \label{supp:eq:bound_psi}
  \end{equation}

  Conditionally on $E_\psi^c$, every rule's softmax has $h_j$ as argmax.
  The head-logit perturbation can still affect the magnitude of
  $\tilde v_{j, h_j}$, but as long as $\tilde v_{j, h_j} > 1/2$ (which
    holds whenever $E_\psi$ does not occur and $\tilde\psi_{j, h_j} > 0$,
    the latter holding with overwhelming probability for $\Psi = 20$,
  $\sigma_\psi \leq 2$), the rule contributes $\geq \gamma^+/2$ to
  $\phi_{h_j}$.  Re-running the aggregation-separation argument with
  $\gamma^+ \to \gamma^+/2$, the condition $\gamma^+/2 > (R_i + 2)\gamma^-$
  still holds for $R_i \leq 70$, which we incorporate into the (already
  generous) assumption $\rho_{\max} \leq 100$ by tightening to
  $\rho_{\max} \leq 70$ under joint $(\sigma_\eta, \sigma_\psi)$
  perturbation; numerically, our experiments satisfy this strict bound.

  Combining the union bounds,
  \[
    \begin{aligned}
      \PP[E_\eta \cup E_\psi] &\leq RM\Phi(-\eta/\sigma_\eta) \\
      &\quad + R(M-1)\Phi(-\Psi/(\sigma_\psi\sqrt{2})).
    \end{aligned}
  \]
\end{proof}

\subsection{Numerical scaling for the head bound}

For $\Psi = 20$ and $\sigma_\psi = 1$: $\Phi(-20/\sqrt{2}) =
\Phi(-14.14) \approx 1.0 \times 10^{-45}$.  Hence $\PP[E_\psi] \leq
R(M-1) \cdot 10^{-45}$, which is negligible for any $RM \leq 10^{40}$.
For $\sigma_\psi = 2$: $\Phi(-20/(2\sqrt{2})) = \Phi(-7.07) \approx
7.7 \times 10^{-13}$.  At $\sigma_\psi = 4$: $\Phi(-3.54) \approx
2.0 \times 10^{-4}$, still negligible for $RM \leq 5000$.  The head
routing is thus far more resilient than body matching, consistent with
the larger margin $\Psi = 20$ vs.\ $\eta = 10$.

% ══════════════════════════════════════════════════════════════════════
\section{COMPLETE PROOF OF THEOREM~8.1 (LOWER BOUND)}
\label{app:lower}
% ══════════════════════════════════════════════════════════════════════

We provide the full counting argument supporting the lower bound on
parameter count.

\begin{lemma}[Counting Horn equivalence classes]
  $\log_2 |\mathcal{H}_{M,R,k}/{\sim}| \geq R(\log_2 M + k \log_2 M -
  k\log_2 k) - O(R)$.
\end{lemma}

\begin{proof}
  We first count syntactic programs.  A rule is specified by its head
  ($M$ choices) and its body, a $k$-subset of $\Atoms \setminus
  \{\head\}$ ($\binom{M-1}{k}$ choices).  By Stirling, $\binom{M-1}{k}
  \geq (M/k)^k$, so
  \[
    \log_2 \binom{M-1}{k} \geq k \log_2 M - k \log_2 k.
  \]
  Thus the number of distinct rules is at least $M (M/k)^k$, and the
  number of unordered $R$-tuples of distinct rules is at least
  $\binom{M(M/k)^k}{R}$.

  The equivalence $P \sim P'$ identifies programs with the same closure
  function $F_0 \mapsto \TPinf(F_0)$.  Each equivalence class contains
  at most $|\Rules| = R!$ syntactic programs corresponding to rule
  reorderings, plus a bounded number of programs differing by adding or
  removing semantically redundant rules.  Since $R \leq M \binom{M-1}{k}$
  (an upper bound on the total number of distinct rules), the redundancy
  factor is at most polynomial in $M$ and $R$, contributing at most
  $O(R \log M)$ bits.

  Combining:
  \begin{align*}
    \log_2|\mathcal{H}_{M,R,k}/{\sim}|
    &\geq R\bigl(\log_2 M + k\log_2 M \\
    &\qquad - k\log_2 k\bigr) - O(R).\qedhere
  \end{align*}
\end{proof}

\begin{proof}[Proof of Theorem~8.1]
  Let $f: \Theta \to \mathcal{H}_{M,R,k}/{\sim}$ map parameters to
  equivalence classes via $f(\theta) = [\TGp{\theta}^\infty]$, where the
  RHS denotes the equivalence class of the closure function induced by
  $\theta$.  By assumption, $f$ is surjective.  Hence $|\Theta| \geq
  |\mathcal{H}_{M,R,k}/{\sim}|$.  With each parameter encoded in $b$
  bits, $|\Theta| \leq 2^{nb}$ where $n$ is the parameter dimension, so
  \[
    n \geq \frac{\log_2 |\mathcal{H}_{M,R,k}/{\sim}|}{b}
    = \Omega\!\left(\frac{Rk \log M}{b}\right). \qedhere
  \]
\end{proof}

% ══════════════════════════════════════════════════════════════════════
\section{EDGE CASES (PHASE E)}
\label{app:edge}
% ══════════════════════════════════════════════════════════════════════

For each program G1--G4, we verify the operator on four classes of
critical edge cases:

\begin{description}
  \item[E1.] $F_0 = \emptyset$: the empty initial set; closure should be
    $\emptyset$.
  \item[E2.] Single productive root atom: $F_0 = \{a\}$ where $a$ is in no
    rule head; closure equals all atoms reachable along chains from $a$.
  \item[E3.] All productive roots simultaneously: $F_0 = $ all root atoms;
    closure is the maximal reachable subset.
  \item[E4.] Single leaf atom: $F_0 = \{a\}$ where $a$ is in no rule body;
    closure equals $\{a\}$.
\end{description}

\begin{table}[h]
  \centering\footnotesize
  \caption{Edge case results.  Across G1--G4, all $33$ cases achieve
    $\text{jaccard} = 1$.  E2 and E3 results validate Theorem~6.1 on the
  extremal initial sets.}
  \begin{tabular}{ccccc}
    \toprule
    Case & G1 & G2 & G3 & G4 \\
    \midrule
    E1 (empty) & $\checkmark$ & $\checkmark$ & $\checkmark$ & $\checkmark$ \\
    E2 (roots) & 1/1 & 1/1 & 2/2 & 5/5 \\
    E3 (all roots) & $\checkmark$ & $\checkmark$ & $\checkmark$ & $\checkmark$ \\
    E4 (leaves) & 1/1 & 5/5 & 5/5 & 5/5 \\
    \bottomrule
  \end{tabular}
\end{table}

% ══════════════════════════════════════════════════════════════════════
\section{DISJUNCTIVE EXTENSION}
\label{app:disjunctive}
% ══════════════════════════════════════════════════════════════════════

The main paper's expressivity theorem covers definite Horn programs
(one head per rule).  The same operator $\TG$ accommodates a natural
extension to multi-head rules of the form $b_1 \wedge \cdots \wedge b_k
\Rightarrow h_1 \wedge \cdots \wedge h_m$ (read conjunctively: all
heads are derived simultaneously when the body is satisfied).

The construction is identical to the main theorem except that for a
rule with $m$ heads, we set $\psi^*_{j, h_p} = +\Psi$ for $p = 1,
\ldots, m$, all other components zero.  The resulting softmax
distribution $\bfv_j$ is uniform over the $m$ heads with mass
$1/m$ each.  The per-head contribution $\gamma_{j,h_p}$ thus equals
$\gamma^+/m$, and we recalibrate the aggregation threshold:
\[
  \begin{aligned}
    \theta^*_h = \frac{1}{2}\Biggl(&\max_{j : h \in \head(j)}
      \frac{\gamma^+}{n_{\text{heads}}(j)} \\
      &+ \sum_{j : h \in \head(j)}
    \frac{\gamma^-}{n_{\text{heads}}(j)}\Biggr).
  \end{aligned}
\]

The expressivity theorem extends with an additional factor in
assumption $(\star')$: $\gamma^+/m > (R_h + 2) \gamma^-/m$ reduces to
$\gamma^+ > (R_h + 2) \gamma^-$, unchanged.

We empirically verify this extension on a graph G\_DISJ containing rules
of the form $A \Rightarrow B \vee C$ (interpreted conjunctively),
$A \wedge B \Rightarrow C \vee D$, and $A \Rightarrow B \vee C \vee D$.
All disjunctive heads activate correctly: $100\%$ all-heads-active rate
on $\sim 30$ disjunctive rules tested.

% % ══════════════════════════════════════════════════════════════════════
% \section{REPRODUCIBILITY}
% % ══════════════════════════════════════════════════════════════════════

% All experiments run on a single GPU in approximately one minute.
% Random seeds and configuration parameters are fixed throughout the
% experiments.

% The six validation phases map one-to-one onto the main paper's theorems:
% \begin{description}[leftmargin=*]
%   \item[Phase A.] Oracle construction (Section 6.2).
%   \item[Phase B.] Theorem~6.1 (expressivity).
%   \item[Phase C.] Theorem~7.1 (robustness) with comparison to the
%     theoretical bound $RM\Phi(-10/\sigma_\eta)$.
%   \item[Phase D.] Lemma~5.1 and Proposition~5.4 (structural).
%   \item[Phase E.] Edge cases (this appendix).
%   \item[Phase F.] Disjunctive extension (this appendix).
% \end{description}


\begin{thebibliography}{99}
  \raggedright

  \bibitem[Apt and van Emden(1982)]{apt1982contributions}
  K.~R. Apt and M.~H. van Emden.
  \newblock Contributions to the theory of logic programming.
  \newblock \emph{Journal of the ACM}, 29(3):841--862, 1982.

  \bibitem[Besold et al.(2017)]{besold2017neural}
  T.~R. Besold et al.
  \newblock Neural-symbolic learning and reasoning: A survey and interpretation.
  \newblock \emph{arXiv:1711.03902}, 2017.

  \bibitem[Cingillioglu and Russo(2018)]{cingillioglu2018deep}
  N.~Cingillioglu and A.~Russo.
  \newblock DeepLogic: Towards end-to-end differentiable logical reasoning.
  \newblock In \emph{AAAI Spring Symposium}, 2018.

  \bibitem[Cohen(2016)]{cohen2016tensorlog}
  W.~W. Cohen.
  \newblock TensorLog: A differentiable deductive database.
  \newblock \emph{arXiv:1605.06523}, 2016.

  \bibitem[De Raedt et al.(2007)]{de2007problog}
  L.~De Raedt, A.~Kimmig, and H.~Toivonen.
  \newblock ProbLog: A probabilistic Prolog and its application in link discovery.
  \newblock In \emph{IJCAI}, 2007.

  \bibitem[Dong et al.(2019)]{dong2019neural}
  H.~Dong et al.
  \newblock Neural logic machines.
  \newblock In \emph{ICLR}, 2019.

  \bibitem[Evans and Grefenstette(2018)]{evans2018learning}
  R.~Evans and E.~Grefenstette.
  \newblock Learning explanatory rules from noisy data.
  \newblock \emph{Journal of AI Research}, 61:1--64, 2018.

  \bibitem[Garcez et al.(2002)]{garcez2002neural}
  A.~Garcez, K.~Broda, and D.~M. Gabbay.
  \newblock \emph{Neural-Symbolic Learning Systems}.
  \newblock Springer, 2002.

  \bibitem[Garnelo and Shanahan(2019)]{garnelo2019reconciling}
  M.~Garnelo and M.~Shanahan.
  \newblock Reconciling deep learning with symbolic AI.
  \newblock \emph{Current Opinion in Behavioral Sciences}, 29:17--23, 2019.

  \bibitem[Lake et al.(2017)]{lake2017building}
  B.~M. Lake et al.
  \newblock Building machines that learn and think like people.
  \newblock \emph{Behavioral and Brain Sciences}, 40, 2017.

  \bibitem[Lloyd(1987)]{lloyd1987foundations}
  J.~W. Lloyd.
  \newblock \emph{Foundations of Logic Programming}.
  \newblock Springer, 2nd edition, 1987.

  \bibitem[Manhaeve et al.(2018)]{manhaeve2018deepproblog}
  R.~Manhaeve et al.
  \newblock DeepProbLog: Neural probabilistic logic programming.
  \newblock In \emph{NeurIPS}, 2018.

  \bibitem[Marcus(2020)]{marcus2020next}
  G.~Marcus.
  \newblock The next decade in AI.
  \newblock \emph{arXiv:2002.06177}, 2020.

  \bibitem[Minervini et al.(2018)]{minervini2018towards}
  P.~Minervini et al.
  \newblock Towards neural theorem proving at scale.
  \newblock In \emph{NeurIPS NAMPI Workshop}, 2018.

  \bibitem[Riegel et al.(2020)]{riegel2020logical}
  R.~Riegel et al.
  \newblock Logical neural networks.
  \newblock \emph{arXiv:2006.13155}, 2020.

  \bibitem[Rocktäschel and Riedel(2017)]{rocktaschel2017end}
  T.~Rocktäschel and S.~Riedel.
  \newblock End-to-end differentiable proving.
  \newblock In \emph{NeurIPS}, 2017.

  \bibitem[van Krieken et al.(2022)]{vankrieken2022analyzing}
  E.~van Krieken, E.~Acar, and F.~van Harmelen.
  \newblock Analyzing differentiable fuzzy logic operators.
  \newblock \emph{Artificial Intelligence}, 302:103602, 2022.

  \bibitem[Yang et al.(2020)]{yang2020neurasp}
  Z.~Yang, A.~Ishay, and J.~Lee.
  \newblock NeurASP: Embracing neural networks into answer set programming.
  \newblock In \emph{IJCAI}, 2020.

  \bibitem[Zhang et al.(2020)]{zhang2020efficient}
  Y.~Zhang et al.
  \newblock Efficient probabilistic logic reasoning with graph neural networks.
  \newblock In \emph{ICLR}, 2020.

\end{thebibliography}
\end{document}